\documentclass{llncs}

\usepackage{subcaption} 
\usepackage[noend]{algpseudocode} 
\usepackage[ruled]{algorithm} 
\usepackage[tbtags]{amsmath} 
\usepackage{amssymb}
\usepackage{booktabs} 
\usepackage{hyperref} 
\usepackage{listings} 
\usepackage{microtype} 
\usepackage[frozencache]{minted} 
\usepackage{llncsdoc}
\usepackage{paralist} 
\usepackage{pgfplotstable} 
\usepackage{pgfplots} 
\usepackage{stmaryrd} 
\usepackage{tikz} 
\usepackage[textsize=footnotesize]{todonotes}
\usepackage{enumitem} 
\usepackage{letltxmacro} 
\usepackage{xspace} 

\usetikzlibrary{graphs, positioning,shapes,backgrounds,patterns}
\usetikzlibrary{intersections}
\usetikzlibrary{calc}
\usetikzlibrary{external}

\tikzstyle{program_node}=[circle,draw=blue!50,thick,minimum size=6mm]
\tikzstyle{large_program_node}=[circle,draw=blue!50,thick,minimum size=15mm]
\tikzstyle{transition}=[->,thick]
\tikzstyle{selected_transition}=[->,very thick,densely dashed, draw=selected]

\definecolor{listinggray}{gray}{0.9}
\definecolor{lbcolor}{rgb}{0.9,0.9,0.9}
\definecolor{selected}{rgb}{0.85,0.12,0.12}

\renewcommand{\vec}[1]{\uppercase{#1}}

\makeatletter\@ifclassloaded{llncs}{
\AtBeginEnvironment{definition}{\let\itshape\relax}
\AtBeginEnvironment{example}{\let\itshape\relax}
}{}\makeatother

\algrenewcommand\Return{\State \algorithmicreturn{} }
\algnewcommand{\LineComment}[1]{\vspace{0.4em}\State \(\triangleright\) #1}
\algnewcommand{\CLineComment}[1]{\State \(\triangleright\) #1}

\makeatletter
\g@addto@macro\normalsize{%
\abovedisplayskip 5.0\p@ \@plus2\p@ \@minus4\p@
\abovedisplayshortskip \z@ \@plus2\p@
\belowdisplayshortskip 3\p@ \@plus2\p@ \@minus2\p@
\belowdisplayskip \abovedisplayskip
}
\g@addto@macro\small{%
\abovedisplayskip 7.5\p@ \@plus2\p@ \@minus4\p@
\abovedisplayshortskip \z@ \@plus2\p@
\belowdisplayshortskip 3\p@ \@plus2\p@ \@minus2\p@
\belowdisplayskip \abovedisplayskip
}
\renewcommand{\paragraph}{%
  \@startsection{paragraph}{4}%
  {\z@}{0.8ex \@plus 0ex \@minus 1ex}{-1em}%
  {\normalfont\normalsize\bfseries}%
}

\def\@IEEEsectpunct{.\ \,}

\newlength{\sectionspace}
\renewcommand\section{\@startsection{section}{1}{\z@}%
                       {-18\sectionspace \@plus -4\sectionspace \@minus -4\sectionspace}%
                       {12\sectionspace \@plus 4\sectionspace \@minus 4\sectionspace}%
                       {\normalfont\large\bfseries\boldmath
                        \rightskip=\z@ \@plus 8em\pretolerance=10000 }}

\makeatother

\newcommand{\model}{\ensuremath{\mathcal{M}}}
\newcommand{\ttbox}[1]{\mbox{\texttt{#1}}}
\newcommand{\keyword}[1]{\textsc{#1}}
\newcommand{\semantics}[1]{\ensuremath{\left\llbracket #1 \right\rrbracket}}

\newcommand{\svcomp}[0]{\textsc{SV-COMP}\xspace}

\newcommand{\rcnf}{\textsc{RCNF}\xspace}

\newcommand{\cfa}{\textsc{CFA}\xspace}

\newcommand{\art}{\textsc{ART}\xspace}
\newcommand{\setof}[1]{\ensuremath{\left \{{#1}\right\}}}
\newcommand{\sequence}[1]{\ensuremath{\left \langle{#1}\right\rangle}}
\newcommand{\tuple}[1]{\ensuremath{\left( #1 \right) }}

\newcommand{\definetool}[2]{\newcommand{#1}{{\small\sc #2}\xspace}}

\definetool{\javasmt}       {JavaSMT}
\definetool{\slicer}        {Slicer}
\definetool{\lpi}           {LPI}

\definetool{\mathsat}       {MathSAT}
\definetool{\optimathsat}   {OptiMathSAT}
\definetool{\zzz}           {Z3}
\definetool{\smtinterpol}   {SMTInterpol}
\definetool{\princess}      {Princess}
\definetool{\cpachecker}    {CPAchecker}
\definetool{\jconstraints}  {jConstraints}
\definetool{\pysmt}         {PySMT}
\definetool{\Houdini}       {Houdini}
\definetool{\pagai}       {PAGAI}

\LetLtxMacro{\oldmissingfigure}{\missingfigure}
\renewcommand{\missingfigure}[2][]{\tikzexternaldisable\oldmissingfigure[{#1}]{#2}\tikzexternalenable}

\LetLtxMacro{\oldtodo}{\todo}
\renewcommand{\todo}[2][]{\tikzexternaldisable\oldtodo[#1]{#2}\tikzexternalenable}

\begin{document}

\title{Formula Slicing: Inductive Invariants from Preconditions\thanks{The research leading to these results has received funding from the
\href{http://erc.europa.int/}{European Research Council} under the European
Union's Seventh Framework Programme (FP/2007-2013) / ERC Grant Agreement
nr.~306595 \href{http://stator.imag.fr/}{\mbox{``STATOR''}}.}}

\author{Egor George Karpenkov \and David Monniaux}
\institute{Univ. Grenoble Alpes, VERIMAG, F-38000 Grenoble, France\\
CNRS, VERIMAG, F-38000 Grenoble, France}

\date{\today}

\maketitle

\begin{abstract}
We propose a ``formula slicing'' method for finding inductive invariants.
It is based on the observation that many loops in the program
affect only a small part of the memory,
and many invariants which were valid before a loop are still valid after.

Given a precondition of the loop,
obtained from the preceding program fragment,
we weaken it until it becomes inductive.
The weakening procedure is guided by counterexamples-to-induction given by an
SMT solver.
Our algorithm applies to programs with arbitrary loop structure,
and it computes the strongest invariant in an abstract domain of weakenings of
preconditions.
We call this algorithm ``formula slicing'', as it effectively performs
``slicing'' on formulas derived from symbolic execution.

We evaluate our algorithm on the device driver benchmarks
from the International Competition on Software Verification (SV-COMP),
and we show that it is competitive with the state-of-the-art verification
techniques.
\end{abstract}

\section{Introduction}
\label{sec:introduction}

In automated program verification, one crucial task is establishing
\emph{inductive invariants} for loops: properties that hold initially,
and also by induction for any number of execution steps.

Abstract-interpretation-based approaches restrict the class of expressible
invariants to a predefined \emph{abstract domain}, such as
intervals, octagons, or convex polyhedra (all of which can only express convex properties).
Any \emph{candidate invariants} which can not be expressed in the chosen abstract
domain get over-approximated.
Traditionally, this restriction applies at all program locations,
but approaches such as \emph{path focusing}~\cite{path_focusing} limit the
precision loss only to loop heads,
representing program executions between the loop-heads \emph{precisely}
using first-order formulas.

This is still a severe restriction:
if a property flows from the beginning of the program to a loop head,
and holds inductively after,
but is not representable within the chosen abstract domain, it is discarded.
In contrast, our idea exploits the insight that many loops in the program
affect only a small part of the memory,
and many invariants which were valid before the loop are still valid.

Consider finding an inductive invariant for the motivating example in Fig.~\ref{fig:slicing_example}.
Symbolic execution up to the loop-head can precisely express
all reachable states:
\begin{equation}
i = 0 \land (p \neq 0 \implies x \geq 0) \land (p = 0 \implies x < 0)
\label{eq:reachable_before}
\end{equation}
Yet abstraction in a numeric convex domain at the loop head
yields $i = 0$,
completely losing the information that $x$ is positive iff $p \neq 0$.
Observe that this information loss is not \emph{necessary},
as the sign of $x$ stays invariant under the multiplication by a positive constant
(assuming mathematical integers for the simplicity of exposition).
To avoid this loss of precision, we develop a ``formula slicing'' algorithm which
computes inductive \emph{weakenings} of propagated formulas,
allowing to propagate the formulas representing inductive invariants
\emph{across} loop heads.
In the motivating example, formula slicing computes an inductive weakening of
the initial condition in Eq.~\ref{eq:reachable_before}), which is
$(p \neq 0 \implies x \geq 0) \land (p = 0 \implies x < 0)$,
and is thus true at every iteration of the loop.
The computation of inductive weakenings is performed by iteratively filtering
out conjuncts falsified by \emph{counterexamples-to-induction},
derived using an SMT solver.
In the motivating example, transition $i=1$ from $i=0$ falsifies the constraint $i=0$,
and the rest of the conjuncts are inductive.

\begin{figure}[t]
        \begin{minted}{c}
int x = input(), p = input();
if (p)
    assume(x >= 0);
else
    assume(x < 0);
for (int i=0; i < input(); i++) x *= 2;
        \end{minted}
    \caption{Motivating Example for Finding Inductive Weakenings.}
    \label{fig:slicing_example}
\end{figure}

The formula slicing fixpoint computation algorithm
is based on performing abstract interpretation on the lattice
of conjunctions over a finite set of predicates.
The computation starts with a seed invariant which \emph{necessarily} holds
at the given location
on the first time the control reaches it,
and during the computation it is iteratively weakened until inductiveness.
The algorithm terminates within a polynomial number of SMT calls with the
\emph{smallest} invariant which can be expressed in the chosen lattice.

\paragraph*{Contributions}
We present a novel insight for generating inductive
invariants, and a method for creating a lattice of weakenings from an arbitrary
formula describing the loop precondition using a \emph{relaxed conjunctive
    normal form} (Def.~\ref{def:rcnf}) and best-effort quantifier elimination
(Sec.~\ref{sec:existential_quantification}).

We evaluate (Sec.~\ref{sec:evaluation}) our implementation of the formula
slicing algorithm on the ``Device Drivers'' benchmarks from the
International Competition on Software Verification~\cite{svcomp16},
and we demonstrate that it can successfully verify large,
real-world programs which can not be handled with traditional numeric abstract
interpretation,
and that it is competitive with state of the art techniques.

\paragraph*{Related Work}
\label{sec:related_work}
The \emph{Houdini}~\cite{houdini} algorithm mines the program for a set of predicates,
and then finds the largest inductive subset,
dropping the candidate non-inductive
lemmas until the overall inductiveness is achieved.
The optimality proof for \emph{Houdini} is present in the companion
paper~\cite{houdini_proof}.
A very similar algorithm is used by Bradley et Al.~\cite{proving_inductiveness}
to generate the inductive invariants from negations of the counter-examples to induction.

Inductive weakening based on counterexamples-to-induction can be seen as an
algorithm for performing predicate abstraction~\cite{predicate_abstraction}.
Generalizing inductive weakening to \emph{best abstract postcondition computation}
Reps et al.~\cite{best_transformer} use the weakening approach for computing 
the best abstract transformer for any finite-height domain, which we also
perform in Sec.~\ref{sec:inductive_postcondition}.

Generating inductive invariants from a number of heuristically generated lemmas
is a recurrent theme in the verification field.
In \emph{automatic abstraction}~\cite{spacer} a set of predicates is found for the simplified
program with a capped number of loop iterations, and is filtered
until the remaining invariants are inductive for the original, unmodified
program.
A similar approach is used for synthesizing bit-precise invariants by
Gurfinkel et Al.~\cite{frankenbit}.

The complexity of the inductive weakening and that of the related template
abstraction problem are analyzed by Lahiri and Qadeer~\cite{Lahiri_Qadeer_2009}.

\paragraph*{Overview}
We introduce the necessary background in Sec.~\ref{sec:background} and the
weakening algorithm in Sec.~\ref{sec:cti_weakening}.
We define the space of all used weakenings in Sec.~\ref{sec:weakening_space}.
We develop the formula slicing algorithm for applying inductive
weakening to real programs in Sec.~\ref{sec:cpa},
we describe our implementation and the required optimizations and improvements
in Sec.~\ref{sec:implementation},
and we conclude with the empirical evaluation on the \svcomp dataset in
Sec.~\ref{sec:evaluation}.

\section{Background}
\label{sec:background}

\subsection{Logic Preliminaries}
\label{sec:background_logic}

We operate over first-order, existentially quantified
logic formulas within an efficiently decidable theory.
A set of all such formulas over free variables in $X$ is denoted by
$\mathcal{F}(X)$.
Checking such formulas for satisfiability is NP-hard,
but with modern SMT (\emph{satisfiability modulo theories}) solvers
these checks can often be performed very fast.

A formula is said to be an \emph{atom} if it does not contain logical connectives
(e.g. it is a comparison $x \leq y$ between integer variables),
a \emph{literal} if it is an atom or its negation,
 and a \emph{clause} if it is a disjunction of literals.
A formula is in \emph{negation normal form} (NNF) if negations are
applied only to atoms,
and it is in \emph{conjunctive normal form} (CNF) if it is a conjunction of clauses.
For a set of variables $X$, we denote by $X'$ a set where the prime symbol
was added to all the elements of $X$.
With $\phi[a_1/a_2]$ we denote the formula $\phi$ after all free occurrences of the
variable $a_1$ have been replaced by $a_2$.
This notation is extended to sets of variables: $\phi[X/X']$ denotes the formula
$\phi$ after all occurrences of the free variables from $X$ were replaced with
corresponding free variables from $X'$.
For brevity, a formula $\phi[X/X']$ may be denoted by $\phi'$.
We use the brackets notation to indicate what free variables can occur in a
formula: e.g. $\phi(X)$ can only contain free variables in $X$.
The brackets can be dropped if the context is obvious.

A formula $\phi(X)$, representing a set of program states, is said to be
\emph{inductive} with respect to a formula $\tau(X \cup X')$, representing a
\emph{transition}, if Eq.~\ref{eq:inductiveness} is valid:

\begin{equation}\label{eq:inductiveness}
  \phi(X) \land \tau(X \cup X') \implies \phi'(X')
\end{equation}

That is, all transitions originating in $\phi$ end up in $\phi'$.
We can query an SMT solver for the inductiveness of $\phi(X)$ with
respect to $\tau(X \cup X')$ using the constraint in Eq.~\ref{eq:inductiveness_checking},
which is unsatisfiable iff $\phi(X)$ is inductive.
    \begin{equation}\label{eq:inductiveness_checking}
        \phi(X) \land \tau(X \cup X') \land \lnot \phi'(X')
    \end{equation}

For a quantifier-free formula $\phi$ inductiveness checking is co-NP-complete.
However, if $\phi$ is existentially quantified,
the problem becomes $\Pi^p_2$-complete.
For efficiency, we shall thus restrict inductiveness checks to quantifier-free
formulas.

\subsection{Program Semantics and Verification Task}
\label{sec:background_semantics}

\begin{definition}[\textsc{CFA}]
    \label{def:cfa}
    A control flow automaton is a tuple $(\mathit{nodes}, \mathit{edges}, n_0, X)$,
    where $\mathit{nodes}$ is a set of program control states, modelling the
    program counter, $n_0 \in \mathit{nodes}$ is a program starting point,
    and $X$ is a set of program variables.
    Each edge $e \in \mathit{edges}$ is a tuple $\tuple{a, \tau(X \cup X'), b}$,
    modelling a possible transition,
    where $\setof{a, b} \subseteq \mathit{nodes}$,
    and $\tau(X \cup X')$ is a formula defining the semantics of a transition over
    the sets of input variables $X$ and output variables $X'$.
\end{definition}

A non-recursive program in a C-like programming language can be trivially converted to a
\cfa by inlining functions,
replacing loops and conditionals with guarded \texttt{goto}s,
and converting guards and assignments to constraints over
input variables $X$ and output variables $X'$.

A \emph{concrete data state} $m$ of a \cfa is a variable assignment
$X \to \mathbb{Z}$
which assigns each variable an integral value.\footnote{The restriction to
    integers is for the simplicity of
    exposition, and is not present in the implementation.}
The set of all concrete data states is denoted by $\mathcal{C}$.
A set $r \subseteq \mathcal{C}$ is called a
\emph{region}.
A formula $\phi(X)$ defines a region $S$ of all states which it models
($S \equiv \{  c \mid c \models \phi \}$).
A set of all formulas over $X$ is denoted by $\mathcal{F}(X)$.
A \emph{concrete state} $c$ is a tuple $\tuple{m, n}$ where $m$ is a concrete
data state, and $n \in \mathit{nodes}$ is a control state.
A \emph{program path} is a sequence of concrete states $\sequence{c_0, \ldots, c_n}$
such that for any two consecutive states
$c_i = \tuple{m_i, n_i}$ and $c_{i+1} = \tuple{m_{i+1}, n_{i+1}}$ there exists
an edge $\tuple{n_{i}, \tau, n_{i+1}}$ such that
$m_i(X) \cup m_{i+1}(X') \models \tau(X \cup X')$.
A concrete state $s_i = \tuple{m, n}$, and the contained node $n$,
are both called \emph{reachable} iff there exists a program path which
contains~$s_i$.

A \emph{verification task} is a pair $\tuple{P, n_e}$ where $P$ is a
\textsc{CFA} and $n_e \in \mathit{nodes}$ is an \emph{error node}.
A verification task is \emph{safe} if $n_e$ is not reachable.
Safety is traditionally decided by finding a \emph{separating} inductive
invariant: a mapping from program locations to regions which is closed under the
transition relation and does not contain the error state.

\subsection{Invariant and Inductive Invariant}
\label{sec:inductive_invariant}

A set of concrete states is called a \emph{state-space},
and is defined using a mapping from nodes to regions.
A mapping $I: \mathit{nodes} \to \mathcal{F}(X)$ is an \emph{invariant} if it
contains \emph{all} reachable states,
and an \emph{inductive invariant} if it is closed under the transition relation:
that is, it satisfies the conditions for \emph{initiation} and
\emph{consecution}:

\begin{equation}\label{eq:cfa_inductiveness}
\begin{split}
    \text{Initiation: } & I(n_0) = \top \\
    \text{Consecution: }& \text{for all edges $\tuple{a, \tau, b} \in
        \mathit{edges}$, for all $X, X'$} \\
        & I(a)(X) \land \tau(X \cup X') \implies (I(b))'(X')
\end{split}
\end{equation}

Intuitively, the initiation condition dictates that the initial program state at $n_0$
(arbitrary contents of memory) is covered by $I$,
and the consecution condition dictates that under all transitions $I$ should map
into itself.
Similarly to Eq.~\ref{eq:inductiveness_checking}, the consecution condition in
Eq.~\ref{eq:cfa_inductiveness} can be verified by checking one constraint for
unsatisfiability using SMT for each edge in a \textsc{CFA}.
This constraint is given in Eq.~\ref{eq:cfa_consecution_check},
which is unsatisfiable for each edge $\tuple{a, \tau, b} \in \mathit{edges}$
iff the consecution condition holds for $I$.

\begin{equation}\label{eq:cfa_consecution_check}
    I(a)(X) \land \tau(X \cup X') \land \lnot (I(b))'(X')
\end{equation}

\subsection{Abstract Interpretation Over Formulas}
\label{sec:background_abstract_interpretation}
Program analysis by abstract interpretation~\cite{abstract_interpretation} searches
for inductive invariants in a given \emph{abstract domain}: the class of
properties considered by the analysis
(e.g. upper and lower bounds on each numeric variable).
The run of abstract interpretation effectively \emph{interprets} the program in
the given \emph{abstract domain}, performing operations on the elements of
an abstract domain instead of concrete values
(e.g. the interval $x \in [1, 2]$ under the transition
\ttbox{x += 1} becomes $x \in [2, 3]$).

We define the abstract domain $\mathcal{D} \equiv 2^{\mathcal{L}} \cup \{ \bot \}$
to be a powerset of the set of formulas $\mathcal{L} \subseteq \mathcal{F}(X)$
with an extra element $\bot$ attached.
A \emph{concretization} of an element $d \in \mathcal{D}$ is a 
conjunction over all elements of $d$, or a formula $\mathit{false}$ for $\bot$.

Observe that $\mathcal{D}$ forms a complete lattice
by using set operations of intersection and union as meet and join operators
respectively, and using \emph{syntactical} equality for comparing individual
formulas.
The syntactic comparison is an over-approximation
as it does not take the formula semantics into account.
However, this comparison generates a complete lattice of height $\|\mathcal{L}\| + 2$.

\subsection{Large Block Encoding}
\label{sec:large_block_encoding}

The approach of large block encoding~\cite{large_block_encoding} for model checking,
and the approach of path focusing~\cite{path_focusing} for abstract interpretation
are based on the observation that by \emph{compacting} a control flow and reducing a
number of abstraction points,
analysis precision and sometimes even analysis performance can be greatly
improved.
Both approaches utilize SMT solvers for performing abstraction afterwards.

A simplified version of compaction is possible by applying the
following two rules to a \textsc{CFA} until a fixed point is reached:
\begin{itemize}
    \item Two consecutive edges $\tuple{a, s_1, b}$ and $\tuple{b, s_2, c}$ with
        no other existing edge entering or leaving $b$ get replaced by a new edge
        $(a, \exists \hat{X} \ldotp s_1[X'/\hat{X}] \land s_2[X/\hat{X}], c)$.
    \item Two parallel edges $\tuple{a, s_1, b}$ and $\tuple{a, s_2, b}$ get
        replaced by $\tuple{a, s_1 \lor s_2, c}$.
\end{itemize}
In our approach, this pre-processing is used on the \textsc{CFA} obtained from the
analyzed program.

\section{Counterexample-to-Induction Weakening Algorithm}
\label{sec:cti_weakening}

The approaches~\cite{houdini,proving_inductiveness,spacer,frankenbit}
mentioned in Sec.~\ref{sec:related_work} are all based on using counterexamples
to induction for filtering the input set of candidate lemmas.
For completeness, we restate this approach in
Alg.~\ref{alg:counterexample_driven}.

In order to perform the weakening without syntactically modifying $\phi$ during
the intermediate queries,
we perform \emph{selector variables} annotation:
we replace each lemma $l_i \in \phi$ with a disjunction $s_i \lor l_i$,
using a fresh boolean variable $s_i$.
Observe that if all selector variables are assumed to be false the annotated
formula $\phi_{\text{annotated}}$ is equivalent to $\phi$,
and that assuming any individual selector $s_i$ is equivalent to removing
(replacing with $\top$) the corresponding lemma $l_i$ from $\phi$.
Such an annotation allows us to make use of \emph{incrementality}
support by SMT solvers, by using the \emph{solving with assumptions} feature.

Alg.~\ref{alg:counterexample_driven} iteratively checks input formula $\phi$
for inductiveness using Eq.~\ref{eq:inductiveness_checking}
(line~\ref{alg:counterexample_driven:termination}).
The solver will either report that the constraint is unsatisfiable,
in which case $\phi$ is inductive,
or provide a counterexample-to-induction represented by a model $\model(X \cup
X')$ (line~\ref{alg:counterexample_driven:model}).
The counterexample-driven algorithm uses $\model$ to find the
set of lemmas which should be removed from $\phi$,
by removing the lemmas modelled by $\model$ in $\lnot \phi'$
(line~\ref{alg:counterexample_driven:drop}).
The visualization of such a filtering step for a formula $\phi$ consisting of
two lemmas is given in Fig.~\ref{fig:cex_weakening_filtering}.

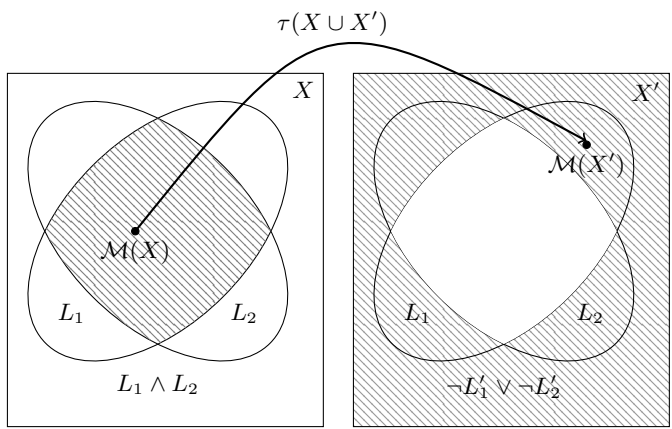
\begin{figure}[t]
    \centering
    \tikzpicturedependsonfile{diagrams/non_inductive_region.tex}
    \begin{tikzpicture}[decoration={}]

\begin{scope}
    \draw[decorate, rotate=45] (0, 0) ellipse (60pt and 35pt) node[anchor=north west, below left=1.2]{$L_1$};
    
    \draw[decorate, rotate=45] (0, 0) ellipse (35pt and 60pt) node[anchor=north, below right=1.2]{$L_2$}
      node[anchor=north, below=1.8]{$ L_1 \land L_2$};
    
    \begin{scope}
    	\clip[decorate, rotate=45] (0, 0) ellipse (60pt and 35pt);
    	\fill[pattern color=gray, decorate, pattern=north west lines, rotate=45] (0, 0) ellipse (35pt and 60pt);
    \end{scope}
    
    \draw[fill](-0.3, 0) circle (1.4pt) node[anchor=north] {$\mathcal{M}(X)$};
   
    \draw[] (-2, -2.6) rectangle (2.2, 2.1) node[anchor=north east] {$X$};

\end{scope}

\begin{scope}[xshift=4.6cm]

    \draw[pattern color=gray!80, pattern=north west lines] (-2, -2.6) rectangle (2.2, 2.1) node[anchor=north east] {$X'$};

    \draw[decorate, rotate=45] (0, 0) ellipse (60pt and 35pt) node[anchor=north west, below left=1.2]{$L_1$};
    
    \draw[decorate, rotate=45] (0, 0) ellipse (35pt and 60pt) node[anchor=north, below right=1.2]{$L_2$}
      node[anchor=north, below=1.8]{$\lnot L_1' \lor \lnot L_2'$};

    \begin{scope}
    	\clip[decorate, rotate=45] (0, 0) ellipse (60pt and 35pt);
    	\fill[fill=white, rotate=45] (0, 0) ellipse (35pt and 60pt);
    \end{scope}

    \draw[fill](1.1, 1.15) circle (1.4pt) node[anchor=north] {$\mathcal{M}(X')$};

\end{scope}

    \draw[->,thick] (-0.3, 0) .. controls (2.2,3.1) and (2.2,3.1) .. node[anchor=south, above]{$\tau(X \cup X')$} (5.7, 1.2);

\end{tikzpicture}
    \caption{Formula $\phi(X) \equiv L_1(X) \land L_2(X)$
        is tested for inductiveness under $\tau(X \cup X')$.
        Model $\model$ identifies a counter-example to induction.
        From $\model \models \lnot L_2'(X')$ we know that the
        lemma $L_2$ has to be dropped.
        As weakening progresses, the shaded region in the left box is growing, while
        the shaded region in the right box is shrinking, until there are no more
        counterexamples to induction.}
    \label{fig:cex_weakening_filtering}
\end{figure}

\begin{algorithm}[t]
\begin{algorithmic}[1]
\State \textbf{Input: } Formula $\phi(X)$ to weaken in \textsc{RCNF},
                        transition relation $\tau(X \cup X')$
\State \textbf{Output: } Inductive $\hat{\phi} \subseteq \phi$
\LineComment{Annotate lemmas with selectors,
    $S$ is a mapping from selectors to lemmas they annotate}.
\State $S, \phi_{\text{annotated}} \gets$ \Call{Annotate}{$\phi$}
    \label{alg:counterexample_driven:annotation}
\State $T \gets$ SMT solver instance
\State $\mathit{query} \gets \phi_{\text{annotated}} \land \tau \land \lnot
\phi_{\text{annotated}}'$
    \label{alg:counterexample_driven:query}
\State Add $\mathit{query}$ to constraints in $T$
\State $\mathit{assumptions} \gets \emptyset$
\State $\mathit{removed} \gets \emptyset$

\LineComment{In the beginning, all of the lemmas are present}
\ForAll{$\tuple{\mathit{selector}, \mathit{lemma}} \in S$}
    \State $\mathit{assumptions} \gets \mathit{assumptions} \cup
            \setof{\lnot \mathit{selector}}$
\EndFor

\While{$T$ is satisfiable with $\mathit{assumptions}$}
        \label{alg:counterexample_driven:termination}
    \State $\model \gets$ model of $T$
        \label{alg:counterexample_driven:model}

    \State $\mathit{assumptions} \gets \emptyset$
    \ForAll{$\tuple{\mathit{selector}, \mathit{lemma}} \in S$}
            \label{alg:counterexample_driven:iteration}
        \If{$\model \models \lnot \mathit{lemma}'$ or
                $\mathit{lemma}'$ is \emph{irrelevant} to satisfiability}
                    \label{alg:counterexample_driven:chosen}
            \LineComment{$\mathit{lemma}$ has to be removed.}
            \State $\mathit{assumptions} \gets \mathit{assumptions} \cup
                \{\mathit{selector}\}$
            \State $\mathit{removed} \gets \mathit{removed} \cup
                \{\mathit{lemma}\}$ \label{alg:counterexample_driven:drop}
        \Else
            \State $\mathit{assumptions} \gets \mathit{assumptions} \cup
                \setof{\lnot \mathit{selector}}$
        \EndIf
    \EndFor
\EndWhile
\LineComment{Remove all lemmas which were filtered out}
\Return $\phi[\mathit{removed}/\top]$
\end{algorithmic}
\caption{Counterexample-Driven Weakening.}
\label{alg:counterexample_driven}
\end{algorithm}

As shown in related literature~\cite{houdini_proof},
Alg.~\ref{alg:counterexample_driven} terminates with the \emph{strongest}
possible weakening within the linear number of SMT calls with respect to 
$\|\phi_{\text{annotated}}\|$.

\subsection{From Weakenings to Abstract Postconditions}
\label{sec:inductive_postcondition}

As shown by Reps et Al.~\cite{best_transformer}, the inductive weakening
algorithm can be generalized for the abstract postcondition computation for any
finite-height lattice.

For given formulas $\psi(X)$, $\tau(X \cup X')$, and $\phi(X)$ consider the problem of
finding a weakening $\hat{\phi} \subseteq \phi$, such that all feasible transitions from $\psi$
through $\tau$ end up in $\hat{\phi}$.
This is an abstract postcondition of $\psi$ under $\tau$ in the lattice of all weakenings
of $\phi$ (Sec.~\ref{sec:background_abstract_interpretation}).
The problem of finding it is very similar to the problem of
finding an inductive weakening,
as similarly to Eq.~\ref{eq:inductiveness_checking},
we can check whether a given weakening of $\phi$
is a postcondition of $\psi$ under $\tau$ using Eq.~\ref{eq:inductive_postcondition},

\begin{equation}\label{eq:inductive_postcondition}
    \psi(X) \land \tau(X \cup X') \land \lnot \phi_\text{annotated}'(X')
\end{equation}

Alg.~\ref{alg:counterexample_driven} can be adapted for finding the
\emph{strongest} postcondition in the abstract domain of weakenings of the input
formula with very minor modifications.
The required changes are accepting an extra parameter $\psi$,
and changing the queried constraint (line~\ref{alg:counterexample_driven:query})
to Eq.~\ref{eq:inductive_postcondition}.
The found postcondition is indeed strongest~\cite{best_transformer}.

\section{The Space of All Possible Weakenings}
\label{sec:weakening_space}

We wish to find a \emph{weakening} of a set of states represented by $\phi(X)$,
such that it is inductive under a given transition $\tau(X \cup X')$.
For a single-node \textsc{CFA} defined by initial condition $\phi$
and a loop transition $\tau$ such a weakening would constitute an
\emph{inductive invariant} as by definition of weakening it satisfies the
initial condition and is inductive.

We start with an observation that for a formula in \keyword{NNF} replacing any
subset of literals with $\top$ results in an over-approximation,
as both conjunction and disjunction are monotone operators.
E.g. for a formula $\phi \equiv (l_a \land l_b) \lor l_c$ such possible
weakenings are $\top$, $l_b \lor l_c$, and $l_a \lor l_c$.

The set of weakenings defined in the previous paragraph is redundant, as it does
not take the formula structure into account --- e.g. in the given example
if $l_c$ is replaced with $\top$ it is irrelevant what other literals are
replaced, as the entire formula simplifies to $\top$.
The most obvious way to address this redundancy is to convert $\phi$ to
\keyword{CNF} and to define the set of all possible weakenings as conjunctions
over the subsets of clauses in $\phi_\text{CNF}$.
E.g. for the formula $\phi \equiv l_a \land l_b \land l_c$ possible weakenings
are $l_a \land l_b$, $l_b \land l_c$, and $l_a \land l_c$.
This method is appealing due to the fact that for a set of lemmas the
\emph{strongest} (implying all other possible inductive weakenings)
inductive subset can be found using a linear number of SMT
checks~\cite{proving_inductiveness}.
However (Sec.~\ref{sec:background_logic})
polynomial-sized \textsc{CNF} conversion (e.g. Tseitin encoding) requires
introducing existentially quantified boolean variables which
make inductiveness checking $\Pi^p_2$-hard.

The arising complexity of finding inductive
weakenings is inherent to the problem: in fact, the problem of finding
\emph{any} non-trivial ($\neq \top$) weakening within the search space described
above is $\Sigma^p_2$-hard (see proof in Appendix~\ref{sec:complexity_proof}).

Thus instead we use an over-approximating set of weakenings,
defined by all possible subsets of lemmas present in $\phi$ after the
conversion to \emph{relaxed conjunctive normal form}.

\begin{definition}[Relaxed Conjunctive Normal Form (RCNF)]
    A formula $\phi(X)$ is in \emph{relaxed conjunctive normal form} if it is a
    conjunction of quantifier-free formulas (lemmas).
    \label{def:rcnf}
\end{definition}

For example, the formula $\phi \equiv l_a \land (l_b \lor (l_c \land l_d))$ is in
\keyword{RCNF}.
The over-approximation comes from the fact that non-atomic parts of the formula
are grouped together: the only possible non-trivial weakenings for
$\phi$ are $l_a$ and $l_b \lor (l_c \land l_d)$, and it is
impossible to express $l_a \land (l_b \lor l_c)$ within the search space.

We may abuse the notation by treating $\phi$ in \keyword{RCNF} as a
set of its conjuncts,
and writing $l \in \phi$ for a lemma $l$ which is an argument of the parent
conjunction of $\phi$, or $\phi_1 \subseteq \phi_2$ to indicate that all
lemmas in $\phi_1$ are contained in $\phi_2$, or $\|\phi\|$ for the number of
lemmas in $\phi$.
For $\phi$ in \textsc{RCNF} we define a set of all possible \emph{weakenings} as
conjunctions over all sets of lemmas contained in $\phi$.
We use an existing, optimal counter-example based algorithm in order to find
the \emph{strongest} weakening of $\phi$ with respect to $\tau$ in
the next section.

A trivially correct conversion to a relaxed conjunctive normal is to convert an input
formula $\phi$ to a conjunction $\bigwedge \setof{\phi}$.
However, this conversion is not very interesting, as it gives rise to a very
small set of weakenings: $\phi$ and $\top$.
Consequently, with such a conversion,
if $\phi$ is not inductive with respect to the transition of interest,
no non-trivial weakening can be found.
On the other extreme, $\phi$ can be converted to \keyword{CNF} explicitly
using associativity and distributivity laws, giving rise to a very large set of
possible weakenings.
Yet the output of such a conversion is exponentially large.

We present an algorithm which converts $\phi$ into a
polynomially-sized conjunction of lemmas.
The following rules are applied recursively until a
fixpoint is reached:
\begin{description}
    \item[Flattening] All nested conjunctions are flattened.
        E.g. $a \land (b \land c) \mapsto a \land b \land c$.
    \item[Factorization]
        When processing a disjunction over multiple conjunctions we find and
        extract a common factor.
        E.g. $(a \land b) \lor (b \land c) \mapsto b \land (a \lor c)$.
    \item[Explicit expansion with size limit] A disjunction $\bigvee L$,
        where each $l \in L$ is a conjunction,
        can be converted to a conjunction over disjunctions over all elements in
        the cross product over $L$.
        E.g. $(a \land b) \lor (c \land d)$ can be converted
        $(a \lor c) \land (a \lor d) \land (b \lor c) \land (b \lor d)$.

        Applying such an expansion results in an exponential
        blow-up, but we only perform it if the resulting formula size is smaller
        than a fixed constant, and we limit the expansion depth to one.
\end{description}

\paragraph*{Eliminating Existentially Quantified Variables}
\label{sec:existential_quantification}
The formulas resulting form large block encoding
(Sec.~\ref{sec:large_block_encoding})
may have intermediate (neither input nor output), existentially bound variables.
In general, existential quantifier elimination (with e.g. Fourier-Motzkin) is
exponential.
However, for many cases such as simple deterministic assignments, existential
quantifier elimination is easy: e.g. $\exists t \ldotp x' = t + 3 \land t = x + 2$ can be
trivially replaced by $x' = x + 5$ using substitution.

We use a two-step method to remove the quantified variables:
we run a best-effort pattern-matching approach,
removing the bound variables which can be eliminated in polynomial time,
and in the second step we drop all the lemmas
which still contain the existentially bound variables.
The resulting formula is an over-approximation of the original one.

\section{Formula Slicing: Overall Algorithm}
\label{sec:cpa}

We develop the \emph{formula slicing} algorithm in order to apply
the inductive weakening approach for generating inductive
invariants in large, potentially non-reducible programs with nested loops.

``Classical'' Houdini-based algorithms consist of two steps: \emph{candidate} lemmas
generation, followed by counterexample-to-induction-based filtering.
However, in our case candidate lemmas representing
postconditions depend on previous filtering steps,
and careful consideration is required in order to generate \emph{unique} candidate
lemmas which do not depend on the chosen iteration order.

\paragraph*{Abstract Reachability Tree}
In order to solve this problem we use abstract reachability tree~\cite{blast} (\art) as a
main datastructure for our algorithm.
For the simplicity of notation we introduce the projection function $\pi_i$,
which projects the $i^\text{th}$ element of the tuple.
An \art describes the current invariant candidate processed by the analysis
for a fixed \cfa $\tuple{\mathit{nodes}, \mathit{edges}, n_0, X}$, and is
defined by a set of nodes $T$.
Each node $t \in T$ is a triple, consisting of a \cfa node
$n \in \mathit{nodes}$, defining which location $t$
corresponds to, an abstract domain element $d \in \mathcal{D}$,
defining the reachable state space at $t$,
and an optional backpointer $b \in (T \cup \{\emptyset\})$,
defining the tree structure.
The tree topology has to be consistent with the structure of the underlying
\cfa: node $a \in T$ can have a backpointer to the node $b \in T$
only if there exists an edge $(\pi_1(a), \_, \pi_1(b))$ in the \cfa.
The starting tree node $t_0$ is $(n_0, \top, \emptyset)$.

An \art is \emph{sound} if the output of each transition
over-approximates the strongest postcondition:
that is, for each node $t \in T$ with non-empty backpointer $b = \pi_3(t)$,
an edge $e = (\pi_1(b), \tau, \pi_1(t))$ must exist in $\mathit{edges}$, and the abstract domain
element associated with $t$ must over-approximate the strongest post-condition
of $b$ under $\tau$.
Formally, the following must hold:
$\exists X \ldotp \semantics{\pi_2(b)} \land \tau \implies \semantics{\pi_2(t)}'$
(recall that priming is a renaming operation $[X/X']$).
A node $b \in T$ is \emph{fully expanded} if for all edges
$(\pi_1(t), \tau, n) \subseteq \mathit{edges}$ there exists a node $t \in T$,
where $\pi_1(t) = n$, and $\pi_2(t)$ over-approximates the strongest
post-condition of $\pi_2(b)$ under $\tau$.
A node $(a, d_1, \_)$ \emph{covers} another node $(a, d_2, \_)$ iff
$\semantics{d_2} \implies \semantics{d_1}$.
A sound labelled \art where all nodes are either fully expanded or
covered represents an inductive invariant.

The transfer relation for the formula slicing is given in
Alg.~\ref{alg:formula_slicing}.
In order to generate a successor for an element $(n_a, d, b)$,
and an edge $(n_a, \tau, n_b)$ we
first traverse the chain of backpointers up the tree.
If we can find a ``sibling'' element $s$ where $\pi_1(s) = n_a$\footnote{
    In the implementation, the \emph{sibling} is defined by a
    combination of callstack, \cfa node and loopstack.} by following the
backpointers, we weaken $s$ until
inductiveness (line~\ref{alg:formula_slicing:weaken}) relative to the new
incoming transition $\tau$, and return that as a postcondition.
Such an operation effectively performs widening~\cite{abstract_interpretation}
to enforce convergence.
Alternatively, if no such sibling exists, we convert $\exists X \ldotp \land \tau$
to \rcnf form (line~\ref{alg:formula_slicing:to_rcnf}),
and this becomes a new element of the abstract domain.

The main fixpoint loop performs the following calculation:
for every leaf in the tree which is not yet expanded or covered,
all successors are found
using the transfer relation defined in Alg.~\ref{alg:formula_slicing},
and for each newly created element, coverage relation is checked against all elements in
the same partition.
A simplified version of this standard fixpoint iteration on \art is given in
Alg.~\ref{alg:formula_slicing:overall}.

\begin{algorithm}[t]
\begin{algorithmic}[1]
    \State \textbf{Input: } \cfa $(\mathit{nodes}, \mathit{edges}, n_0, X)$
    \LineComment{Expanded.}
    \State $E \gets \emptyset$
    \LineComment{Covered.}
    \State $C \gets \emptyset$
    \State $t_0 \gets (n_0, \top, \emptyset)$
    \State $T \gets \{t_0\}$
    \While{$\exists t \in (T \setminus E \setminus C)$}
        \LineComment{Expand.}
        \ForAll{edge $e \in \mathit{edges}$ where $\pi_1(e) = \pi_1(t)$}
            \State $T \gets T \cup \{$ \Call{TransferRelation}{$e, t$} $\}$
        \EndFor
        \State $E \gets E \cup \{ t \}$
        \LineComment{Check Coverage.}
        \ForAll{$t_1 \in (T \setminus C)$ where $\pi_1(t_1) = \pi_1(t)$}
            \If{$\semantics{\pi_2(t_1)} \implies \semantics{\pi_2(t)}$}
                \State $C \gets C \cup \{ t_1 \}$
            \EndIf
            \If{$\semantics{\pi_2(t)} \implies \semantics{\pi_2(t_1)}$}
                \State $C \gets C \cup \{ t \}$
            \EndIf
        \EndFor
    \EndWhile
\end{algorithmic}
\caption{Formula Slicing: Overall Algorithm}
\label{alg:formula_slicing:overall}
\end{algorithm}

Observe that our algorithm has a number of positive features.
Firstly, because our main datastructure is an \art,
in case of a counterexample we get a \emph{path} to a property violation
(though due to abstraction used, not all taken transitions are necessarily
feasible, similarly to the \emph{leaping counterexamples} of LoopFrog~\cite{loopfrog}).
Secondly, our approach for generating initial candidate invariants ensures
uniqueness, even in the case of a non-reducible \cfa.

As a downside, tree representation may lead to the exponential state-space
explosion (as a single node in a \cfa may correspond to many nodes in an \art).
However, from our experience in the evaluation (Sec.~\ref{sec:evaluation}),
with a good iteration order (stabilizing inner components first~\cite{wto})
this problem does not occur in practice.

\begin{algorithm}[t]
\begin{algorithmic}[1]
    \Function{TransferRelation}{edge $e \equiv (n_a, \tau, n_b)$,
                                state $t \equiv (n_a, d, b)$}
        \State sibling $s \gets$ \Call{FindSibling}{$b, n_0$}
        \If{$s \neq \emptyset$}
            \LineComment{Abstract postcondition of $d$ under
                $\tau$ in weakenings of $s$
                    (Sec.~\ref{sec:inductive_postcondition}).}
            \label{alg:formula_slicing:weaken}
            \State $e \gets$ \Call{Weaken}{$d$, $\tau \land n_b$, $s$}
        \Else
            \label{alg:formula_slicing:to_rcnf}
            \LineComment{Convert the current invariant candidate to \rcnf.}
            \State $e \gets$ \Call{ToRCNF}{$\semantics{d} \land \tau$}
        \EndIf
        \Return $(n_b, e, t)$
    \EndFunction

    \Function{FindSibling}{state $b$, \cfa node $n$}
        \If{$\pi_1(b) = n$}
            \Return $b$
        \ElsIf{$\pi_3(b) = \emptyset$}
            \Return $\emptyset$
        \Else
            \Return \Call{FindSibling}{$\pi_3(b), n$}
        \EndIf
    \EndFunction
\end{algorithmic}
\caption{Formula Slicing: Postcondition Computation.}
\label{alg:formula_slicing}
\end{algorithm}

\subsection{Example Formula Slicing Run}
Consider running formula slicing on the program in Fig.~\ref{ex:nested_loops},
which contains two nested loops.
The corresponding edge encoding is given in Eq.~\ref{eq:edge_encoding}:

\begin{equation}\label{eq:edge_encoding}
\begin{split}
    \tau_1 \equiv& x' = 0 \land y' = 0 \land
        (p' = 1 \land s' \lor p' = 2 \land \lnot s') \\
    \tau_2 \equiv& x' = x + 1 \land c' = 100 \\
    \tau_3 \equiv& (\lnot (p \neq 1 \land p \neq 2) \lor
         (p \neq 1 \land p \neq 2 \land c' = 0)) \\
        & \land y' = y + 1 \land p' = p \\
    \tau_4 \equiv& x' = x \land y' = y \land p' = p \land c'=c
\end{split}
\end{equation}

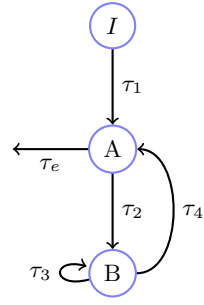
\begin{figure}[t]
    \tikzpicturedependsonfile{diagrams/nested_cfa.tex}
    \begin{subfigure}[t]{0.5\textwidth}
        \centering
        \begin{minted}{c}
int p, c, s=nondet(), x = 0, y = 0;
p = s ? 1 : 2;
while (nondet()) { // A(X)
    x++;
    c = 100;
    while (nondet()) { // B(X)
        if (p != 1 && p != 2) {
            c = 0;
        }
        y++;
    }
    assert(c == 100);
}
assert((s && p == 1) || (!s && p == 2));
        \end{minted}
        \vspace{-1mm}
    \end{subfigure} %
    \hspace{3mm} %
    \begin{subfigure}[t]{0.5\textwidth}
        \centering
        \vspace{-1mm}
        \begin{tikzpicture}[auto, scale=1, transform shape]
    \node (i) [program_node] {$I$};
    \node (a) [program_node, below=1 of i] {A};
    \node (out) [left=1 of a] {};
    \node (b) [program_node, below=1 of a] {B};

    \path [transition] (i) edge node {$\tau_1$} (a)
        (a) edge node  {$\tau_e$} (out)
        (a) edge node {$\tau_2$} (b)
        (b) edge [loop left] node {$\tau_3$} (b);
    \draw[transition] (b) to [bend right=90] node[anchor=west] {$\tau_4$} (a);
\end{tikzpicture}
    \end{subfigure}
    \caption{Example Program with Nested Loops: Listing and \cfa.}\label{ex:nested_loops}
\end{figure}

Similarly to Eq.~\ref{eq:inductiveness_checking}, we can check candidate
invariants $A(X), B(X)$ for inductiveness by posing an SMT query shown in
Eq.~\ref{eq:multi_loop_inductiveness_checking}.
The constraint in Eq.~\ref{eq:multi_loop_inductiveness_checking} is
unsatisfiable iff $\{ A: A(X), B: B(X) \}$ is an inductive invariant
(Sec.~\ref{sec:inductive_invariant}).

\begin{equation}\label{eq:multi_loop_inductiveness_checking}
    \exists X \cup X' \bigvee \begin{aligned}
    \tau_1(X') &\land \lnot A(X') \\
    A(X) \land \tau_2(X \cup X') &\land \lnot B(X') \\
    B(X) \land \tau_3(X \cup X') &\land \lnot B(X') \\
    B(X) \land \tau_4(X \cup X') &\land \lnot A(X')
\end{aligned}
\end{equation}

Eq.~\ref{eq:multi_loop_inductiveness_checking} is unsatisfiable iff \emph{all}
of the disjunction arguments are unsatisfiable, and hence the checking can be
split into multiple steps, one per analyzed edge.
Each postcondition computation (Alg.~\ref{alg:formula_slicing})
either generates an initial seed invariant candidate,
or picks one argument of Eq.~\ref{eq:multi_loop_inductiveness_checking},
and weakens the right hand side until the constraint becomes unsatisfiable.
Run of the formula slicing algorithm on the example is given below:

\begin{itemize}
    \item Traversing $\tau_1$, we get the initial candidate invariant \\
        $I(A) \gets \bigwedge \setof{x = 0, y = 0, p = 1 \lor p = 2, s \implies p = 1}$.
    \item Traversing $\tau_2$, the candidate invariant for $B$ becomes \\
        $I(B) \gets \bigwedge \setof{x = 1, y = 0, p = 1 \lor p = 2, s \implies p = 1,
        c=100 }$.
    \item After traversing $\tau_3$, we weaken the candidate invariant $I(B)$ by
        dropping the lemma $y=0$ which gives rise to the
        counterexample to induction ($y$ gets incremented).
        The result is $\bigwedge \setof{x = 1, p = 1 \lor p = 2, s \implies p = 1,
        c=100 }$, which is inductive under $\tau_3$.
    \item The edge $\tau_4$ is an identity, and the postcondition computation
        results in lemmas $x=0$ and $y=0$ dropped from $I(A)$, resulting in
        $\bigwedge \{y = 0, p = 1 \lor p = 2, s \implies p = 1\}$.
    \item After traversing $\tau_2$, we obtain the weakening of $I(A)$ by
        dropping the  lemma $x=1$ from $I(B)$, resulting in 
        $\bigwedge \setof{p = 1 \lor p = 2, s \implies p = 1, c=100 }$.
    \item Finally, the iteration converges, as all further postconditions are
        already covered by existing invariant candidates.
        Observe that the computed invariant is sufficient for proving the
        asserted property.
\end{itemize}

\section{Implementation}
\label{sec:implementation}

We have developed the \textsc{Slicer} tool, which runs the formula slicing algorithm on an input C program.
\textsc{Slicer} performs inductive weakenings using the \textsc{Z3}~\cite{z3} SMT solver,
and best-effort quantifier elimination using the \texttt{qe-light} Z3 tactic.
The source code is integrated inside the open-source verification framework
CPAchecker~\cite{cpachecker},
and the usage details are available at~\url{http://slicer.metaworld.me}.
Our tool can analyze a verification task (Sec.~\ref{sec:background_semantics})
by finding an inductive invariant and reporting \texttt{true}
if the found invariant \emph{separates} the initial state from the error property,
and \texttt{unknown} otherwise.

We have implemented the following optimizations:
\begin{description}
    \item[Live Variables] We precompute live variables,
        and the candidate lemmas generated during \rcnf conversion
        (Alg.~\ref{alg:formula_slicing},
            line~\ref{alg:formula_slicing:to_rcnf})
        which do not contain live variables are discarded.
    \item[Non-Nested Loops] When performing the inductive weakening
        (Alg.~\ref{alg:formula_slicing}, line~\ref{alg:formula_slicing:weaken})
        on the edge $\tuple{N, \tau, N}$ we annotate and weaken the candidate
        invariants on both sides (without modifications described in
        Sec.~\ref{sec:inductive_postcondition}), and we cache the fact that the
        resulting weakening is inductive under $\tau$.
    \item[\textsc{CFA} Reduction] We pre-process the input \cfa and we remove all nodes
        from which there exists no path to an error state.
\end{description}

\subsection{Syntactic Weakening Algorithm}
\label{sec:syntactic_weakening}

A syntactic-based approach is possible as a faster and less
precise alternative which does not require SMT queries.
For an input formula $\phi(X)$ in \textsc{RCNF},
and a transition $\tau(X \cup X')$,
syntactic weakening returns a subset of lemmas in $\phi$,
which are not \emph{syntactically modified} by $\tau$:
that is, none of the variables are modified or have their address taken.
For example, the lemma $x > 0$ is not syntactically modified by the
transition $y' = y + 1 \land x \geq 1$, but it is modified by $x' = x + 1$.

\section{Experiments and Evaluation}
\label{sec:evaluation}

We have evaluated the formula slicing algorithm on the ``Device Drivers'' category
from the International Competition on Software Verification
(\textsc{SV-COMP})~\cite{svcomp16}.
The dataset consists of $2120$ verification tasks,
of which $1857$ are designated as \emph{correct}
(the error property is unreachable), and the rest admit a counter-example.
All the experiments were performed on
Intel Xeon E5-2650 at 2.00 GHz,
and limits of 8GB RAM, 2 cores, and 600 seconds CPU time per program.
We compare the following three approaches:

\begin{description}
    \item[Slicer-CEX] (rev \texttt{21098}) Formula slicing algorithm running
        counterexample-based weakening (Sec.~\ref{sec:cti_weakening}).
    \item[Slicer-Syntactic] Same, with syntactic weakening
        (Sec.~\ref{sec:syntactic_weakening}).
    \item[Predicate Analysis] (rev \texttt{21098}) Predicate abstraction with
        interpolants~\cite{lazy_abstraction},
        as implemented inside CPAchecker~\cite{predicateCPA}.
        We have chosen this approach for comparison as it represents
        state-of-the-art in model checking, and was found especially suitable for
        analyzing device drivers.
    \item[PAGAI]~\cite{pagai} (git hash \texttt{e44910})
        Abstract interpretation-based tool,
        which implements the path focusing~\cite{path_focusing} approach.
\end{description}
Unabridged experimental results are
available at~\url{http://slicer.metaworld.me}.

In Tab.~\ref{tab:results_table} we show overall precision and performance
of the four compared approaches.
As formula slicing is over-approximating,
it is not capable of finding counterexamples,
and we only compare the number of produced safety proofs.

From the data in the table we can see that predicate analysis
produces the most correct proofs.
This is expected since it can generate new predicates,
and it is \emph{driven} by the target property.
However, formula slicing and abstract interpretation have much less timeouts,
and they do not require target property annotation, making them more suitable for
use in domains where a single error property is not available
(advanced compiler optimizations, multi-property verification, and boosting
another analysis by providing an inductive invariant).
The programs verified by different approaches are also different,
and formula slicing verifies $22$ programs predicate analysis could not.

The performance of the four analyzed approaches is shown in the quantile
plot in Fig.~\ref{fig:quantile_plot}.
The plot shows that predicate analysis is considerably more time consuming
than other analyzed approaches.
Initially, \textsc{PAGAI} is much faster than other tools,
but around $15$ seconds it gets overtaken by both slicing approaches.
Though the graph seems to indicate that \textsc{PAGAI} overtakes slicing again
around $100$ seconds, in fact the bend is due to out of memory errors.

The quantile plot also shows that the time taken to perform inductive
weakening does not dominate the overall analysis time for formula slicing.
This can be seen from the small timing difference between
the syntactic and counterexample-based approaches,
as the syntactic approach does not require querying the SMT solver in order to
produce a weakening.

Finally, we present data on the number of SMT calls required for computing
inductive weakenings in Fig.~\ref{fig:cex_iterations}.
The distribution shows that the overwhelming majority of weakenings can be
found within just a few SMT queries.

\begin{table}[t]
    \centering
    \tikzpicturedependsonfile{results/results_table_compiled.tex}
\begin{tabular}{@{}lrrrr@{}}
    \toprule
    Tool
        & \# proofs
        & \# incorrect
        & \# timeouts
        & \# memory outs
        \\
    \midrule
    
        Slicer-CEX
            & 1253
            & 0
            & 475
            & 0
            \\
    
        Slicer-Syntactic
            & 1166
            & 0
            & 407
            & 0
            \\
    
        Predicate Analysis
            & 1301
            & 0
            & 657
            & 0
            \\
    
        PAGAI
            & 1214
            & 3
            & 409
            & 240
            \\
    
    \bottomrule
\end{tabular}

    \vspace{2mm}
    \caption{Evaluation results.
        The ``\# incorrect'' column shows the number of safety proofs the tool
        has produced where the analyzed program admitted a counterexample.}
    \label{tab:results_table}
\end{table}

\begin{figure}[t]
    \tikzpicturedependsonfile{results/quantile_plot.tex}
    \tikzpicturedependsonfile{results/quantile_plot.dat}
    \tikzpicturedependsonfile{results/cex_iterations_histogram.tex}
    \tikzpicturedependsonfile{results/histogram.dat}

    \hspace{-3mm}%
    \begin{subfigure}[t]{0.5\textwidth}
        \hspace{-5mm}\begin{tikzpicture}[scale=0.8]
    \pgfplotstableread{results/quantile_plot.dat}\comparisontable
    \begin{axis}[xmode=normal, axis equal=false,  ymode=log,
        xlabel={Programs},
        ylabel={Wall Time (s)},
        enlarge x limits = 0.05,
        enlarge y limits = 0.07,
        legend entries={
            Slicer-CEX,
            Slicer-Syntactic,
            Predicate Analysis,
            PAGAI
        },
        legend style={
            draw=none,
            fill=none,
            cells={anchor=west},
            legend pos={north west},
            font={\footnotesize},
            at={(0.01, 0.99)},
        }
        ]
        \addplot+[mark repeat=20] table[x index=0, y index=1] {\comparisontable};
        \addplot+[mark repeat=20, purple, mark phase=5] table[x index=0, y index=2] {\comparisontable};
        \addplot+[mark repeat=20, teal, mark phase=0, mark=+] table[x index=0, y index=3] {\comparisontable};
        \addplot+[mark repeat=20, mark phase=5] table[x index=0, y index=4] {\comparisontable};
    \end{axis}
\end{tikzpicture}
        \vspace{-1mm}
        \caption{Quantile plot showing performance of the compared approaches.
            Shows analysis time for each benchmark,
            where the data series are sorted by time separately for each tool.
            For readability, the dot is drawn for every $20^{\text{th}}$ program,
            and the time is rounded up to one second.
        }
        \label{fig:quantile_plot}
    \end{subfigure} %
    \hspace{1mm} %
    \begin{subfigure}[t]{0.5\textwidth}
        \begin{tikzpicture}[scale=0.8]
    \pgfplotstableread{results/histogram.dat}\blah
    \begin{axis} [
        axis equal=false,
        ymode=log,
        ybar interval=1,
        restrict x to domain=0:70,
        enlargelimits=false,
        xlabel=\# SMT queries,
        ylabel=Required for \# weakenings,
        xtick=data,
        x tick label style={font=\footnotesize, anchor=north east},
        ]
    \addplot table {\blah};

    \end{axis}
\end{tikzpicture}
        \vspace{-1mm}
        \caption{Distribution of the number of iterations of inductive weakening
            (Sec.~\ref{sec:cti_weakening}) required for convergence across all
            benchmarks.
            Horizontal axis represents the number of SMT calls required for
            convergence of each weakening, and vertical axis represents the count of
            the number of such weakenings.
        }
        \label{fig:cex_iterations}
    \end{subfigure}
\end{figure}

\section{Conclusion and Future Work}
\label{sec:conclusion}

We have proposed a ``formula slicing'' algorithm for efficiently finding
potentially disjunctive inductive invariants in programs,
which performs abstract interpretation in the space of weakenings over the
formulas representing the ``initial'' state.
We have demonstrated that it could verify many programs other approaches
could not, and that the algorithm can be run on real programs.

The motivation for our approach is addressing the limitation
of abstract interpretation which forces it to perform abstraction after each
analysis step,
which often results in a very rough over-approximation.
Thus we believe our method is well-suited for augmenting numeric abstract
interpretation.

As with any new inductive invariant generation technique, a possible future work
is investigating whether formula slicing can be used for increasing the performance and
precision of other program analysis techniques, such as $k$-induction, predicate
abstraction or property-directed reachability.
An obvious approach would be feeding the invariants generated by formula slicing
to a convex analysis running abstract interpretation or policy
iteration~\cite{LPI}.

Furthermore, the inductive weakening approach could also be used for the generalization of
the $k$-induction algorithm over multiple properties.
If we check a set of properties $P$ for inductiveness
under the loop transition $\tau$,
and $\bigwedge P$ is not inductive,
the weakening can find the largest inductive subset.

\paragraph{Acknowledgements} The authors wish to thank Grigory Fedyukovich and
Alexey Bakhirkin for proof-reading and providing valuable feedback,
and the anonymous reviewers for their helpful suggestions.

\bibliographystyle{IEEEtran}
{
    \bibliography{library}

\begin{thebibliography}{10}
\providecommand{\url}[1]{#1}
\csname url@samestyle\endcsname
\providecommand{\newblock}{\relax}
\providecommand{\bibinfo}[2]{#2}
\providecommand{\BIBentrySTDinterwordspacing}{\spaceskip=0pt\relax}
\providecommand{\BIBentryALTinterwordstretchfactor}{4}
\providecommand{\BIBentryALTinterwordspacing}{\spaceskip=\fontdimen2\font plus
\BIBentryALTinterwordstretchfactor\fontdimen3\font minus
  \fontdimen4\font\relax}
\providecommand{\BIBforeignlanguage}[2]{{%
\expandafter\ifx\csname l@#1\endcsname\relax
\typeout{** WARNING: IEEEtran.bst: No hyphenation pattern has been}%
\typeout{** loaded for the language `#1'. Using the pattern for}%
\typeout{** the default language instead.}%
\else
\language=\csname l@#1\endcsname
\fi
#2}}
\providecommand{\BIBdecl}{\relax}
\BIBdecl

\bibitem{path_focusing}
D.~Monniaux and L.~Gonnord, ``Using bounded model checking to focus fixpoint
  iterations,'' in \emph{SAS}.\hskip 1em plus 0.5em minus 0.4em\relax Springer,
  2011.

\bibitem{svcomp16}
D.~Beyer, ``Reliable and reproducible competition results with benchexec and
  witnesses ({R}eport on {SV-COMP} 2016),'' in \emph{TACAS}.\hskip 1em plus
  0.5em minus 0.4em\relax Springer, 2016.

\bibitem{houdini}
C.~Flanagan and K.~R.~M. Leino, ``Houdini, an annotation assistant for
  {ESC}/{Java},'' in \emph{FME}, 2001, pp. 500--517.

\bibitem{houdini_proof}
C.~Flanagan, R.~Joshi, and K.~R.~M. Leino, ``Annotation inference for modular
  checkers,'' \emph{Information Processing Letters}, 2001.

\bibitem{proving_inductiveness}
A.~R. Bradley and Z.~Manna, ``Checking safety by inductive generalization of
  counterexamples to induction,'' in \emph{FMCAD}, 2007, pp. 173--180.

\bibitem{predicate_abstraction}
S.~Graf and H.~Sa{\"{\i}}di, ``Construction of abstract state graphs with
  {PVS},'' in \emph{CAV}, 1997, pp. 72--83.

\bibitem{best_transformer}
T.~Reps, M.~Sagiv, and G.~Yorsh, ``Symbolic implementation of the best
  transformer,'' in \emph{VMCAI}, 2004.

\bibitem{spacer}
A.~Komuravelli, A.~Gurfinkel, S.~Chaki, and E.~M. Clarke, ``Automatic
  abstraction in {SMT}-based unbounded software model checking,'' in
  \emph{CAV}, 2013, pp. 846--862.

\bibitem{frankenbit}
A.~Gurfinkel, A.~Belov, and J.~Marques{-}Silva, ``Synthesizing safe bit-precise
  invariants,'' in \emph{TACAS}, 2014, pp. 93--108.

\bibitem{Lahiri_Qadeer_2009}
S.~K. Lahiri and S.~Qadeer, ``Complexity and algorithms for monomial and
  clausal predicate abstraction,'' in \emph{CADE}, 2009, pp. 214--229.

\bibitem{abstract_interpretation}
P.~Cousot and R.~Cousot, ``{Abstract Interpretation: A Unified Lattice Model
  for Static Analysis of Programs by Construction or Approximation of
  Fixpoints},'' in \emph{POPL}, 1977, pp. 238--252.

\bibitem{large_block_encoding}
D.~Beyer, A.~Cimatti, A.~Griggio, M.~E. Keremoglu, and R.~Sebastiani,
  ``Software model checking via large-block encoding,'' in \emph{FMCAD}, 2009,
  pp. 25--32.

\bibitem{blast}
D.~Beyer, T.~A. Henzinger, R.~Jhala, and R.~Majumdar, ``The software model
  checker {Blast},'' \emph{{STTT}}, vol.~9, no. 5-6, pp. 505--525, 2007.

\bibitem{loopfrog}
D.~Kroening, N.~Sharygina, S.~Tonetta, A.~Tsitovich, and C.~M. Wintersteiger,
  ``Loop summarization using abstract transformers,'' in \emph{ATVA}, 2008, pp.
  111--125.

\bibitem{wto}
F.~Bourdoncle, ``Efficient chaotic iteration strategies with widenings,'' in
  \emph{Formal Methods in Programming and Their Applications}, ser. Lecture
  Notes in Computer Science.\hskip 1em plus 0.5em minus 0.4em\relax Springer
  Berlin Heidelberg, 1993, vol. 735, pp. 128--141.

\bibitem{z3}
L.~M. de~Moura and N.~Bj{\o}rner, ``{Z3:} an efficient {SMT} solver,'' in
  \emph{TACAS}, 2008, pp. 337--340.

\bibitem{cpachecker}
D.~Beyer and M.~E. Keremoglu, ``{CPAchecker}: {A} tool for configurable
  software verification,'' in \emph{CAV}, 2011, pp. 184--190.

\bibitem{lazy_abstraction}
K.~L. McMillan, ``Lazy abstraction with interpolants,'' in \emph{CAV}, 2006,
  pp. 123--136.

\bibitem{predicateCPA}
D.~Beyer, M.~E. Keremoglu, and P.~Wendler, ``Predicate abstraction with
  adjustable-block encoding,'' in \emph{FMCAD}, 2010, pp. 189--197.

\bibitem{pagai}
J.~Henry, D.~Monniaux, and M.~Moy, ``{PAGAI:} {A} path sensitive static
  analyser,'' \emph{Electr. Notes Theor. Comput. Sci.}, vol. 289, pp. 15--25,
  2012.

\bibitem{LPI}
E.~G. Karpenkov, D.~Monniaux, and P.~Wendler, ``Program analysis with local
  policy iteration,'' in \emph{VMCAI}.\hskip 1em plus 0.5em minus 0.4em\relax
  Springer, 2016, pp. 127--146.

\bibitem{Stockmeyer_1976}
L.~J. Stockmeyer, ``The polynomial-time hierarchy,'' \emph{Theoretical Computer
  Science}, vol.~3, no.~1, pp. 1--22, 1976.

\end{thebibliography}
}

\clearpage

\appendix
\section{Complexity of Finding a Non-Trivial Inductive Weakening Over Literals}
\label{sec:complexity_proof}

As we have mentioned in Sec.~\ref{sec:weakening_space}, a more expressive space
of weakenings over formulas is to consider replacing any subset of literals with
$\top$ after a \keyword{NNF} conversion.
In this appendix we show that it leads to a number of undesirable properties,
including the absence of \emph{strongest} inductive weakening
(Ex.~\ref{ex:no_strongest}), and $\Sigma^p_2$
complexity for finding any non-trivial inductive weakening
(Thm.~\ref{lemma:weakening_in_Sigma_p_2}).

\begin{example}[No Strongest Inductive Weakening]
    Consider a program over four Boolean variables $a,b,c,d$
    and the transition relation
    $\tau \equiv a \land b \land c \land d \land \neg a' \land b' \land \neg c' \land d'$ 
    (the only possible transition is from $a \land b \land c \land d$ to $\lnot
    a \land b \land \lnot c \land d$).
    Consider finding the weakening of $\phi \equiv (a \land b) \lor (c \land d)$,
    Both the $\{a\}$-weakening ($b \lor (c \land d)$)
    and the $\{c\}$-weakening ($(a \land b) \lor d$) are inductive, but 
    their intersection $(a \land b) \lor (b \land d) \lor (c \land d)$
        (obviously inductive) is not a weakening of $\phi$ and
    there is no inductive weakening stronger than either of these.
    \label{ex:no_strongest}
\end{example}

\begin{theorem}[$\Sigma^p_2$-completeness]
    The problem of deciding, given quantifier-free SMT formulas
    $\phi(X)$ and $\tau(X \cup X')$,
    whether there exists a non-trivial ($\not\equiv \top$) weakening of $\phi$ that is inductive
    with respect to $\tau$ is $\Sigma^p_2$-complete.
    \label{lemma:weakening_in_Sigma_p_2}
\end{theorem}

\begin{proof}[Belonging to $\Sigma^p_2$]
    Let $S$ be some subset of literals of $\phi$.
    Let $\hat{\phi}$ be the weakening of $\phi$ where all literals in $S$ are
    replaced with $\top$.
    Checking that $\hat{\phi}$ is inductive with respect to $\tau$
    is in co-NP, therefore the problem of finding a non-trivial $\hat{\phi}$ is
    in~$\Sigma^p_2$
\end{proof}

We show completeness by constructing from an arbitrary
closed $\exists^* \forall^*$ formula $\psi$
a loop $\tau$ and a precondition $I$ such that the existence of a non-trivial
($\not\equiv \top$) weakening of the precondition is equivalent to the truth of
$\psi$.
Without loss of generality, let $\psi$ have $m$ Boolean variables
$x_0,\dots,x_{m-1}$ bound by the existential quantifier and
$n$ Boolean variables $y_0,\dots,y_{n-1}$ bound by the universal one:
\begin{equation}
    \begin{aligned}
        \psi \equiv &\exists x_0,\dots,x_{m-1} . \\
                    & \forall y_0,\dots,y_{n-1} .
                G(x_0,\dots,x_{m-1},y_0,\dots,y_{n-1})
    \end{aligned}
    \label{eq:psi_definition}
\end{equation}
Let us denote the bitvector $(x_0,\dots,x_{m-1})$ as $\vec{x}$ and the bitvector
$(y_0,\dots,y_{n-1})$ as $\vec{y}$.
Let $\mathit{enc}: \mathbb{B}^m \to [0,2^m-1]$ denote the function
for standard integer encoding of the
$\vec{x}$ bitvector, $x_0$ being the lowest-order bit and $x_{m-1}$ the
highest-order one.
Let $\mathit{succ} : \mathbb{B}^m \setminus \{\top^m\} \to \mathbb{B}^m$ be the successor
function such that $\mathit{enc}(\mathit{succ}(\vec{x})) = 1 +
\mathit{enc}(\mathit{succ}(\vec{x}))$, which is only defined for
non-overflowing values.

Now we define the transition system over the set of boolean variables $\vec{x}$
and the overflow bit $o$.
Let the initial state $I(\vec{x},o)$ be $\vec{x} = \bot \land o = \bot$,
and let the transition relation $\tau(\vec{x},\vec{x'},o,o')$ to be:
\begin{equation}
\begin{aligned}
    &\big(
        \lnot (\forall \vec{y} . G(\vec{x}, \vec{y}))
            \land  \\
                &((\vec{x} \neq \top \land \vec{x'} = \mathit{succ}(\vec{x}) \land o' = o)
                \lor
                (\vec{x} = \top \land o' = \top))
    \big) \\
&\bigvee
    \big(
        \vec{x'} = \vec{x} \land o' = o
    \big)
\end{aligned}
\end{equation}
In plain terms, the transition relation may increment $\vec{x}$ as long as 
it is not overflowing and the guard can be falsified for some $\vec{y}$,
and $\vec{x}$ is forced to stay constant on overflow or
when it reaches some $\vec{\hat{x}}$ such that $\forall \vec{y} . G(\vec{\hat{x}},\vec{y})$.
Initialization and transition relation for the transition system, and the
corresponding program are shown in Fig.~\ref{prog:counter}.

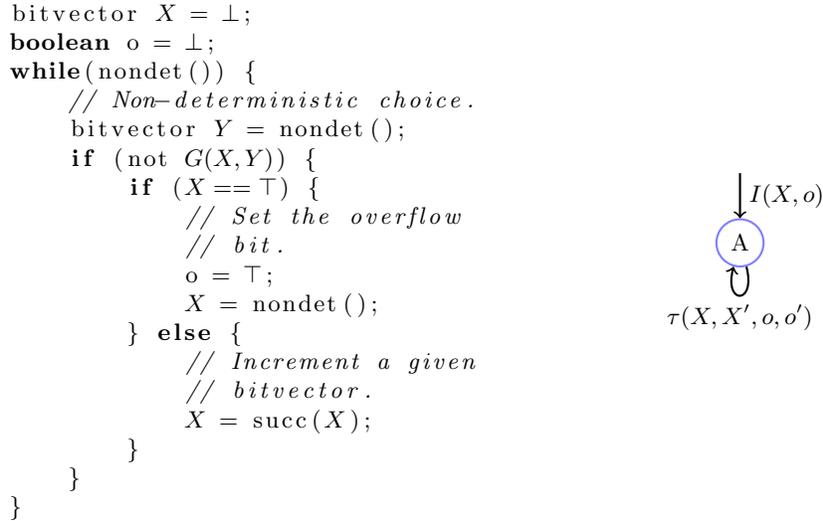
\begin{figure}[t]
    \centering
    \begin{minipage}{0.7\linewidth}
\begin{lstlisting}[language=Java,mathescape=true]

bitvector $\vec{x}$ = $\vec{\bot}$;
boolean o = $\bot$;
while(nondet()) {
    // Non-deterministic choice.
    bitvector $\vec{y}$ = nondet();
    if (not $G(\vec{x},\vec{y})$) {
        if ($\vec{x} == \top$) {
            // Set the overflow
            // bit.
            o = $\top$;
            $\vec{x}$ = nondet();
        } else {
            // Increment a given
            // bitvector.
            $\vec{x}$ = succ($\vec{x}$);
        }
    }
}
\end{lstlisting}
    \end{minipage}
    \begin{minipage}{0.2\linewidth}
\begin{tikzpicture}[auto, scale=1, transform shape]
    \node(entry) [] {};
    \node (a) [program_node, below=0.6 of entry] {A};
    \draw [transition] (entry) to node {$I(\vec{x}, o)$} (a);
    \draw [transition, loop below] (a) to node {$\tau(\vec{x}, \vec{x'}, o, o')$} (a);
\end{tikzpicture}
    \end{minipage}
    \caption{Counter Program and Transition System}
    \label{prog:counter}
    \vspace{-2mm}
\end{figure}

\begin{lemma}
    There exists a non-trivial ($\not\equiv \top$) inductive invariant
    for the program in Fig.~\ref{prog:counter} if and only if $\psi$
    (Eq.~\ref{eq:psi_definition}) is satisfiable.
    \label{lemma:inductiveness_solution}
\end{lemma}

Observe that $\tau$ can be satisfied for all possible values of $\vec{x}$ by
a suitable choice of $\vec{x'}$.
Let $f(\vec{x})$ be the largest (under $\mathit{enc}$) possible value of
$\vec{x'}$ which satisfies $\tau(\vec{x}, \vec{x'}, o, o')$.

\begin{proof}{Sufficient Condition.}
    Assume $\psi$ is satisfiable for some $\hat{\vec{x}}$.
    Then $\hat{\vec{x}}$ is a fixed point under $f$ (as it satisfies $G$ for all
    possible values of $\vec{y}$).
    Consider the set of values defined by $R \equiv \lnot o \land \mathit{enc}(\vec{x}) \leq
    \hat{\vec{x}} \}$.
    It is inductive, since the largest value in $R$ set maps to itself under
    $f$, and all other values map to the ``next'' (under $\mathit{enc}$) value in $R$.
    It is also non-trivial, since the bit $o$ is defined not to be $\top$.
\end{proof}

\begin{proof}{Necessary Condition.}
    Assume there exists a non-trivial inductive invariant for the program in
    Fig.~\ref{prog:counter}.
    At every transition, $\vec{x}$ either stays constant or is incremented by~$1$.
    Since we have assumed the existence of a non-trivial inductive invariant,
    there exists $\vec{\hat{x}}$ such that it is a fixpoint under $f$ and
    $\mathit{enc}(\vec{\hat{x}}) \leq 2^m - 1$ (otherwise the entire state space
    is reachable, and the only possible inductive invariant is $\top$).
    This is only possible if $\forall \vec{y} .
    G(\vec{\hat{x}}, \vec{y})$ (otherwise $\hat{x}$ may be incremented).
    But this is exactly the condition for $\psi$ being satisfiable.
\end{proof}

\begin{corollary}
    For every non-trivial inductive invariant of the program in Fig.~\ref{prog:counter}
    there exists some $\vec{\hat{x}}$ such that
    $\{\vec{x} \mid \mathit{enc}(\vec{x}) < \mathit{enc}(\vec{\hat{x}}) \}$ is
    inductive.
    Furthermore, the reachable state space is exactly all $\vec{x}$ smaller
    (under $\mathit{enc}$) than $\vec{\hat{x}}$,
    and $\{\vec{x} \mid \vec{x} \neq \vec{\hat{x}} \}$ is inductive
    (as the states larger than $\vec{\hat{x}}$ are not reachable).
    \label{corollary:inductive_invariant_shape}
    \vspace{-2mm}
\end{corollary}

Now consider finding inductive (with respect to $\tau$ Fig.~\ref{prog:counter})
weakenings of the following formula $\phi$:
\begin{equation}
    \phi \equiv \bigvee (x_i \land \lnot x_i)
\end{equation}
Each $x_i$ represents $i$'th bit of $\vec{x}$.
Observe that for any $\vec{\hat{x}} \in [0,2^m-1]$, we can weaken $\phi$ to be
equivalent to $\vec{x} \neq \vec{\hat{x}}$,
by making a suitable weakening choice for every $i$'th bit of $\vec{\hat{x}}$
(if the $i$-th bit in $\vec{\hat{x}}$ is $\bot$ we replace $\lnot
x_i$ by $\top$, if it is $\top$ we replace $x_i$ by $\top$).

From Corollary~\ref{corollary:inductive_invariant_shape} we know that for
every non-trivial inductive invariant there exists $\vec{\hat{x}}$, s.t. the set
of all $\vec{x}$ not equal to $\vec{\hat{x}}$ is inductive.
Thus if a non-trivial inductive invariant exists,
there exists a non-trivial inductive weakening of $\phi$.
In Lemma~\ref{lemma:inductiveness_solution} we have shown that
deciding the existence of a non-trivial inductive invariant is
as hard as deciding the satisfiability of an arbitrary $\exists^* \forall^*$ formula
$\psi$, thus deciding an existence of a non-trivial inductive weakening is as hard as well.

\begin{proof}[$\Sigma^p_2$-completeness]
    Membership in $\Sigma^p_2$ is proved in Lemma~\ref{lemma:weakening_in_Sigma_p_2}.
    Reduction from the $\Sigma^p_2$-complete problem is done from deciding the truth of
    $\exists^* \forall^*$ propositional
    formulas~\cite[Th.~4.1]{Stockmeyer_1976}.
    Transforming $G$ into $\tau$ can be done within a logarithmic working space.
\end{proof}

\paragraph*{Relationship to Template Abstraction Complexity}
Lahiri and Qadeer \cite{Lahiri_Qadeer_2009}
consider the problem of \emph{template abstraction}:
given a precondition, a postcondition, a transition relation and a formula
$\phi(C,X)$, $C$ and $X$ being sets of Boolean variables,
check whether an appropriate choice of $C$ makes $\phi$ an inductive invariant.
They show this problem to be $\Sigma^p_2$-complete as well.
Our class of problems is a strict subset of theirs
(our weakening problems can be immediately translated into template abstraction problems,
but not all template abstraction problems correspond to weakenings),
but we still show completeness.

\end{document}